\documentclass[runningheads,anonymous]{llncs}
\usepackage[T1]{fontenc}
\usepackage{graphicx}
\usepackage{CJKutf8}
\usepackage{bm}
\usepackage{amsmath}
\usepackage[ruled,vlined,linesnumbered]{algorithm2e}
\usepackage{algorithmic}
\usepackage{comment}
\usepackage{amssymb}
\usepackage{caption}
\usepackage{subfigure}
\usepackage{booktabs}
\begin{document}
\title{Convex-area-wise Linear Regression and Algorithms for Data Analysis}
%
%

\author{Bohan Lyu\inst{1,2}\orcidID{0009-0005-6462-8942} \and Jianzhong Li\inst{1,2}\orcidID{0000-0002-4119-0571}}
%
\institute{
	\email{18b903024@stu.hit.edu.cn}\\
	\email{lijzh@hit.edu.cn}
	\\
	\inst{1}Harbin Institute of Technology, Harbin, Heilongjiang, China
	\\
	\inst{2}Shenzhen University of Advanced Technology, Shenzhen, Guangdong, China
}

\maketitle              
\begin{abstract}
	This paper introduces a new type of regression methodology named as Convex-Area-Wise Linear Regression(CALR), which separates given datasets by disjoint convex areas and fits different linear regression models for different areas. This regression model is highly interpretable, and it is able to interpolate any given datasets, even when the underlying relationship between explanatory and response variables are non-linear and discontinuous. In order to solve CALR problem, 3 accurate algorithms are proposed under different assumptions. The analysis of correctness and time complexity of the algorithms are given, indicating that the problem can be solved in $o(n^2)$ time accurately when the input datasets have some special features. Besides, this paper introduces an equivalent mixed integer programming problem of CALR which can be approximately solved using existing optimization solvers.

\end{abstract}
\section{Introduction}
Multiple-model linear regression(MLR) is a special and important data analysis method. 
Let $\bm{DS}=\{(\bm{x}_i,y_i)|x_i\in\mathbb{R}^d,y_i \in \mathbb{R},i=1,\cdots,n\}$ be a given dataset where $\bm{x}_i$ is the explanantory variable and $y_i$ is the response variable.
Multiple-model linear regression divides the input datasets into several subsets, and then construct local linear regression model for each subset, such as the example shown in Figure 1 of Appendix A. Such model can approximately represent the non-linear underlying relationship between $y$ and $\bm{x}$. The piecewise linear regression\cite{toriello2012plcf,bemporad2022piecewise}, max-affine regression\cite{ghosh2019max}, PLDC regression\cite{siahkamari2020piecewise}, MMLR\cite{lyu2023efficient} are all typical multiple-model linear regression methods.

Nowadays, real-world datasets, especially big data have the feature that different subsets of a dataset fitting highly different regression models, which is described as \emph{diverse predictor-response variable relationships}(DPRVR) in \cite{dong2015pattern}.
Taking TBI dataset in \cite{TBI2005} as a real-world example. The \emph{mean squared error}(MSE) is 109.2 when one linear regression model is used to model the whole TBI.
However, the $MSE$ is reduced to 12.3, when TBI is divided into 6 disjoint hypercubes and 6 different linear regression models are used individually, while one global $f_0$ is used for the data not belonging to the 6 subsets.
So fitting data using MLR is more suitable than using one single linear or polynomial function in regression tasks.

Besides, MLR model has statistical advantages and high interpretability.
The numeric value of parameters can show the importance of variables, and the belonging information about the confidence coefficients and intervals makes the model more credible in practice.
Therefore, MLR is widely used in research areas requiring high interpretability, such as financial prediction, investment forecasting, biological and medical modelling, etc\cite{book2021lra}.
Most machine learning and deep learning models might be more precise in predicting tasks, but the black-box feature limits their ranges of application\cite{guidotti2018survey}.
These two characteristics make MLR a necessary methodology for nowadays data analysis.

However, there are still shortcomings of the existing multiple-model linear regression methods. The most important ones are described as follows.

1. The dimension $d$ of $\bm{DS}$ cannot be relatively large, since the time complexity of most algorithms constructing MLR models grows exponentially with $d$\cite{toriello2012plcf,wahba1990spline};

2. The subsets being used to construct local models must be hyper cubes or single hyperplane\cite{diakonikolas2020sr,wang1996mtree}. Thus, the accuracy of them is lower when the underlying partition of a given dataset are not in such forms;

3. Some methods need apriori knowledge that is difficult to get, such as the positions of breakpoints\cite{arumugam2003emprr} or the exact number of pieces\cite{diakonikolas2020sr}.

4. The time complexity of the methods is high. Even the state-of-art approximate algorithm, PLDC regression, has the time complexity of $O(d^2 n ^5)$ \cite{siahkamari2020piecewise}.

To overcome the disadvantages and inspired by the TBI dataset, this paper proposes a new multi-model based linear regression method named as Convex-area-wise Linear Regression(CALR). 

Instead of using hypercubes or single hyperplane to separate datasets, CALR considers there is a default linear model $f_0$ lying on the datasets, and there are several disjoint convex areas $C_1,C_2,\cdots,C_M$, which contains subsets of datasets that fitting different local linear models $f_1,\cdots,f_M$, and each convex area $C_i$ can be defined by multiple hyperplanes. An illustration of CALR is shown in Figure 2 of Appendix A.


The major contributions of this paper are concluded as follows.  
\begin{itemize}
	\item A new function class named convex-area-wise linear function(CALF) is formally defined. The important properties and expressive ability of CALF is proved.
	\item The regression problem corresponded with convex-area-wise linear function is formally proposed, named as CALR. The CALR problem has been transformed into a mix-integer programming problem, so that it can be approximately solved by any optimization solvers. 
	\item A exponential time naive algorithm naiveCALR is designed for getting optimal model of CALR problem.
	\item Three algorithms are designed to accurately solve the CALR problem under special assumptions. The correctness and expected time complexity of them are also proved. When $\bm{DS}$ is convex-area separable, the expected time complexity can reach $o(n^2)$, which is irrelevant to the dimension $d$ and it's lower than PLDC regression. 
\end{itemize}

The rest of this paper is organized as follows. Section 2 gives the formal definition and important properties of the new proposed convex-area-wise linear function and the corresponded regression problem. Section 3 gives the design and analysis of three algorithms under different assumptions. Finally, Section 4 concludes the paper.

\section{Preliminaries and Problem Definition}
This section introduces the definition and specialties of the new proposed function class CALF and the formal corresponded regression problem CALR. An optimizing version of CALR is also proposed for users to approximately solve it by existing optimization solvers.

\subsection{Convex-area-wise Linear Function}
Definition of convex-area-wise linear function is given as follows. 

\begin{definition}[Convex-Area-Wise Linear Function]
	$ $
	Suppose that $x\in \Omega \subset \mathbb{R}^d$. Let $H=\{(f_i, C_i)\ | 0 \leq i \leq M\}$ be a set of function-convex area pairs, such that for $i=0,\cdots,M$, $f_i(x)=\beta_i\cdot(1,x)=\beta_{0i} + \beta_{1i} x_1 + \cdots + \beta_{di} x_d$ is a linear function, for $i=1,\cdots,M$, $C_i=\cap_{k=1}^{m_i}\{x|\alpha_{ik}\cdot x + \gamma_{ik} \leq 0, \alpha_{jk} \in \mathbb{R}^d\}$ is a convex area defined by $m_i$ semi-spaces, $C_1,\cdots,C_M$ are disjoint, and $C_0=\Omega-\cup_{i=1}^M C_i$. A function $f_H:\Omega\to\mathbb{R}$ is called a convex-area-wise linear function, CALF for short, if $f_H(x)=\sum_{i=0}^M I_i(x)f_i(x)$, where $I_i(x)=\begin{cases}
		1, x\in C_i \\
		0,x\notin C_i
	\end{cases}$ is the indicative function of $C_i,i=0,\cdots,M$.

	Let $CALF(\Omega)$ be the set of all convex-area-wise linear functions defined on $\Omega\subset \mathbb{R}^d$.
\end{definition}

According to the \textit{Definition 1}, every $f_H\in CALF(\Omega)$ is piece-wise linear on $\Omega$. 
When $x\in C_i,i=1,\cdots,M$, $f_H(x)=f_i(x)$, and $f_H(x)=f_0(x)$ otherwise. 
We say the action scope of $f_i$ is $C_i$. An example of CALF is shown in Figure 2 of Appendix A.

The most important property of CALF is that any finite datasets can be fitted by a CALF function. 
In order to prove the property, the definition of PLDC functions and a corresponded proposition from \cite{siahkamari2020piecewise} are shown without proof. 
Then the characteristics is formally given as the following \emph{Theorem 1}.

\begin{definition}[PLDC function]
	If a function $f$ can be represented as $f(x)=max_{1\leq k\leq K}\{\alpha\cdot x + c_k \} - max_{1\leq k\leq K}\{\beta\cdot x + c'_k \}$, where $K\in \mathbb{Z}^+, x,\alpha_k,\beta_k\in\mathbb{R}^d,c_k,c'_k\in\mathbb{R}$, the $f$ is called as a PLDC function.
\end{definition}
\begin{lemma}\label{lemma1}
	Given any finite data $D=\{(x_i,y_i)|i=1,2,\cdots,n\}$, and $y_i\neq y_j \Rightarrow x_i \neq x_j$, there exists a PLDC function interpolating $D$.
\end{lemma}

\begin{theorem}\label{Theorem 1}
	Given any dataset $D=\{(x_i,y_i)|i=1,2,\cdots,n\}$, where $x_i \in \mathbb{R}^d,y_i\in\mathbb{R}$, $y_i\neq y_j \Rightarrow x_i \neq x_j$, there exists a function $f\in CALF$ fitting $D$.
	
	Furthermore, there exists functions $f$ such that $f\in CALF$ but $f\notin PLDC$. 
\end{theorem}

\begin{proof}
	It's only needed to prove that $PLDC \subset CALF$. Suppose that $f(x)=\max_{1\leq k\leq K}\{\alpha\cdot x + c_k \} - \max_{1\leq k\leq K}\{\beta\cdot x + c'_k \}$ is a PLDC function, and let $f_1(x)=\max_{1\leq k\leq K}\{\alpha\cdot x + c_k \},f_2(x)=\max_{1\leq k\leq K}\{\beta\cdot x + c'_k \}$.
	
	From the definition of $f_1$, it is a continuous function with at most $K$ hyperplanes $\pi_k : \alpha_k \cdot x +c_k - y = 0, k=1,\cdots,K$ in $\mathbb{R}^{d+1}$. 
	Any $\pi_k$ has the boundary represented by $\alpha_k \cdot x+c_k \geq \alpha_j\cdot j+c_j$, which is a semi-space $(\alpha_k -\alpha_j)\cdot x + c_k -c_j \geq 0$, for $j\neq k$.
	
	Thus, the $f_1(x)=\alpha_k \cdot x+ c_k$ when $x\in C_k=\cap_{j\neq k}\{y\in\mathbb{R}^d|(\alpha_k -\alpha_j)\cdot y + c_k -c_j \geq 0\}$ for $k=1,\cdots,K$. 
	$C_k\subset \mathbb{R}^d$ is convex since $C_k$ is the intersection of $K-1$ semi-spaces.
	Similarly, $f_2(x)=\beta_k \cdot x +c'_k$ when $x\in C'_k$, where $C'_k\subset\mathbb{R}^d$ is convex, for $k=1,\cdots,K$.
	
	Then, $f(x)=f_1(x)-f_2(x)$ contains at most $K^2$ pieces of hyperplanes and $f(x)=(\alpha_i-\beta_j)\cdot x + c_k-c'_k$ when $x\in S_i\cap S'_j$, $i,j=1,\cdots,K$. 
	Besides, $C_i \cap C'_j = (\cap_{k\neq i}\{x\in\mathbb{R}^d|(\alpha_i-\alpha_k)\cdot x+c_i-c_k\geq 0\}) \cap (\cap_{k\neq j}\{x\in\mathbb{R}^d|(\beta_j-\beta_k)\cdot x+c'_j-c'_k\geq 0\})$ is a convex area bounded by hyperplanes.
	
	Therefore, $f(x)$ can be represented by a CALF function $f_H$, where $H=\{(g_{ij},C_{ij})|g_{ij}(x)=(\alpha_i-\beta_j)\cdot x +c_i-c'_j, C_{ij}=C_i\cap C'_j,i,j=1,\cdots,K\}$. So if a dataset $D$ can be interpolated by a PLDC function $f$, $D$ can be interpolated by a CALF $f_H$ constructed above.
	
	Combined with \emph{Lemma 1}, the first part of \emph{Theorem 1} is proved.
	
	Further, consider the function$f(x_1,x_2)=
	\begin{cases}
		\frac{1}{2}x_1-1, x_1\leq 0 \\
		-2x_1+1,x_1>0
	\end{cases}$. Obviously, $f(x_1,x_2)$ can be represented by $f_H\in CALF$ with $H={(g_1,C_1),(g_2,C_2)}$ where $g_1(x)=\frac{1}{2}x_1-1,g_2(x)=-2x_1+1$ and $C_1={(x_1,x_2)|x_1\leq 0},C_2={(x_1,x_2)|-x_1\leq \varepsilon}$, $\varepsilon$ is small enough.
	
	However, $f$ can not be represented by any PLDC function. Let $h(x)=h_1(x)-h_2(x)$ be a PLDC function, $h_1(x)$ and $h_2(x)$ are continuous because they are the maximum of finite linear functions \cite{boyd2004convex}. Then $h(x)=h_1(x)-h_2(x)$ is continuous on $\mathbb{R}^2$. Therefore, $PLDC \subsetneq CALF$.
	\qed\end{proof}

\textit{Theorem 1} shows that CALF has stronger expressive ability than PLDC. CALF functions can fit any finite given numerical datasets $D$, even the underline model of $D$ is non-linear and not continuous.

\subsection{Convex-area-wise Linear Regression}
Firstly the definition of traditional linear regression problem is given as follows.
\begin{definition}[Linear Regression Problem]
	$ $
	
	\textbf{Input:} A numerical dataset $\bm{DS}=\{(x_i, $ $y_i)\ | 1 \leq i \leq n \}$, where $x_i \in \mathbb{R}^d, y_i \in \mathbb{R}$, $y_i = f(x_i) + \varepsilon_i$ for some function $f$, $\varepsilon_i \sim N(0,\sigma^2)$, and all $\varepsilon_i$s are independent.
	
	\textbf{Output:} A function $\hat{f}(x)=\bm{\beta} \cdot (1,x)=\beta_0 + \beta_1 x_1 + \cdots + \beta_d x_d$ such that $MSE(\bm{DS},\hat{f})=\frac{1}{n}\sum_{i=1}^n[(y_i-\hat{f}(x_i))^2]$ is minimized. 
\end{definition}

One way to construct optimal linear regression model for $\bm{DS}$ is \emph{pseudo-inverse matrix method}. The time complexity of it is $O(d^2 n + d^3)$ \cite{book2021lra}. So this paper let lr($\bm{DS}$) be any $O(d^2 n + d^3)$ time algorithm to construct the optimal linear function $\hat{f}$ of given datasets $\bm{DS}$.

The p-value of F-test for $\hat{f}$ is a convincing criterion to judge the goodness of $\hat{f}$, which is denoted as $p_F(\hat{f})$ in this paper. The $p_F(f)$ can be calculated in $O(n)$ time\cite{book2021lra}. Generally, linear regression model $\hat{f}\cong f$ when $p_F(f)<0.05$, where $f$ is the underlying model of $\bm{DS}$.
Besides, it's always assumed that $n > d+1$, since a $d$-dimensional linear function can perfectly fit any $n$ data points when $n \leq d+1$.
The details of linear regression and \emph{pseudo-inverse matrix method} is shown in Appendix E.

Similarly, the regression problem corresponded with CALF is defined as follows.

\begin{definition}[Convex-area-wise Linear Regression Problem]
	$ $
	
	\textbf{Input:} A numerical dataset $\bm{DS}=\{(\bm{x}_i, $ $y_i)\ | 1 \leq i \leq n \}$, where $\bm{x}_i \in \mathbb{R}^d, y_i \in \mathbb{R}$, $y_i = f(\bm{x}_i) + \varepsilon_i$ for an $f$, $\varepsilon_i \sim N(0,\sigma^2)$, and all $\varepsilon_i$s are independent identically distributed, Integer $M\geq 0$.
	
	\textbf{Output:} An function $f_H \in CALF$, where $H=\{(\hat{f}_j, \bm{C}_j)\ | 0 \leq j \leq M\}$ such that $MSE(\bm{DS},f_H)=\frac{1}{n}\sum_{i=1}^n(y_i-f_H(\bm{x}_i))^2$ is minimized.

\end{definition}

The problem in \textit{Definition 4} is denoted as CALR problem for short. The decision problem of CALR is given as follows.

\begin{definition}[Decision problem of CALR]
	$ $
	
	\textbf{Input:} A numerical dataset $\bm{DS}=\{(\bm{x}_i, $ $y_i)\ | 1 \leq i \leq n \}$, integer $M\geq 0$, positive $B>0$.
	
	\textbf{Output:} If there exists a function $f_H \in CALF$, where $H=\{(\hat{f}_j, \bm{C}_j)\ | 0 \leq j \leq M\}$ such that $MSE(\bm{DS},f_W)=\frac{1}{n}\sum_{i=1}^n(y_i-f_W(\bm{x}_i))^2 < B$.
\end{definition}

Given $\bm{DS}$ and $f_H\in CALF$, one can calculate $\frac{1}{n}\sum_{i=1}^n(y_i-f_H(\bm{x}_i))^2$ in $O(n)$ time and compare it with $B$ in $O(1)$ time. Therefore, the decision problem of CALR is in NP.

\subsection{The Optimizing Version of CALR}

In subsection 2.2, the definition of CALR problem is given. However, we still need to rewrite it as an optimizing problem to solve it using optimization solvers.

Considering MSE as the optimizing objective, $l_2$ loss function can be used for the problem. Let the number of local models be smaller than $M$, and the numbers of hyperplanes bounding each convex area be smaller than $K$, the objective function is $l(w)=\sum_{i=1}^n(\beta_0\cdot x_i + c_0 + \sum_{j=1}^M I_{ij}(\beta_j\cdot x_i + c_j) -y_i)^2$, where $\beta_j \in \mathbb{R}^d, c_j \in \mathbb{R},j=0,\cdots,M$. 
In this objective function, $I_{ij}$ is the indicator variable for $(x_i,y_i)$ of convex area $C_j$, that is $I_{ij}=\begin{cases} 1, x_i\in C_j \\ 0, x_i\notin C_j \end{cases}$.
Suppose $C_j=\cap_{k=1}^K\{y|\alpha_{jk}\cdot y + \gamma_{jk}\leq0\}$, let $I_{ijk}\in\{-1,1\}$ satisfy $I_{ijk}(\alpha_{jk}\cdot x_i+\gamma_{jk})\leq \tau$, where $\tau<0$ and $|\tau|$ is small enough. 
Then, $x_i\in C_j$ if and only if $\sum_{k=1}^K I_{ijk}=K$.
Besides, it's necessary to let $\sum_{j=1}^MI_{ij}\leq 1$ to make sure $C_1,\cdots,C_M$ are disjoint. 

Thus, the optimizing version of CALR can be given as follows.

\begin{definition}[Optimizing CALR problem]
	$$\min_{w} l(w)=\sum_{i=1}^n(\beta_0\cdot x_i + c_0 + \sum_{j=1}^M I_{ij}(\beta_j\cdot x_i + c_j) -y_i)^2$$
	$$s.t. \begin{cases}
		I_{ijk}(1-I_{ijk})=0 \\
		\sum_{j=1}^M I_{ij} \leq 1 \\
		I_{ij}=\prod_{k=1}^K I_{ijk} \\
		(I_{ijk}-\frac{1}{2})(\alpha_{jk}\cdot x_i+\gamma_{jk})\leq \tau \\
		i=1,\cdots,n;j=1,\cdots,M;k=1,\cdots,K \\
	\end{cases} $$
	where $w=(\beta,c,\alpha,\gamma,I)$ is the vector of optimizing variables, $D=\{(x_i,y_i)|i=1,\cdots,n\}$, $M,K\in\mathbb{Z}^+,\tau<0$ is given.
\end{definition}

This optimizing problem has $(d+1)(K+1)M$ variables and $n(M(2K+1)+1)$ nonlinear constraints.
It's a mix-integer programming problem, which makes it an NP-hard problem\cite{boyd2004convex}. Using optimizing solvers could only get a approximate solution. 

\section{Algorithm and Analysis}

Section 2.3 gives the optimizing version of CALR problem which can only be solved approximately. This section gives 3 accurate algorithms for the problem with different assumptions. 

Firstly, a sub-algorithm cac($\bm{DS},D$) is proposed to construct a convex-area which separated $D$ from $\bm{DS}$. Given $\bm{DS}=\{x_i\in\mathbb{R}^d|i=1,\cdots,n\}$, $D\subset \bm{DS}$ such that $|D|=n'$, cac($\bm{DS},D$) returns $\emptyset$ if the convex hull of $D$ contains $x\in \bm{DS}-D$, returns $C=\cap_{k=1}^K\{z\in\mathbb{R}^d|\alpha_k\cdot z+\gamma_k\leq 0\}$ such that $D\subset C$ and $C\cap \bm{DS}-D=\emptyset$ if there are no $x\in \bm{DS}-D$ such that $x$ is in the convex hull of $D$.

The main idea of cac($\bm{DS},D$) is to transfer the problem into a linear programming problem. For any $x_0\in \bm{DS}-D$, consider the following problem denoted as $LP(x_0,D)$:
$$\min_{w} l(w)=0$$
$$s.t. w\cdot (x_i-x_0) < 0 x_i \in D$$
where $w=(w_1,\cdots,w_d)$ is the coefficient of a hyperplane passing through $x_0$ in $\mathbb{R}^d$. 
According to the definition, $LP(x_0,D)$ is a linear programming problem in $d$-dimension with $n'$ constraints. Gauss-Seidel Method can solve it with $O(n')$ when $d$ is not very large\cite{xu2018hybrid}.
Trivially, any feasible solution of $LP(x_0,D)$ is a hyperplane that separates $x_0$ from $D$. Suppose that $\alpha_k = w_k, \gamma_k = w_k \cdot x_k$ for $x_k \in \bm{DS}-D$, where $w_k$ is one solution of $LP(x_k,D)$, $C=\cap_{k=1}^K\{z\in\mathbb{R}^d|\alpha_k\cdot z+\gamma_k\leq 0\}$ is a convex area separating $D$ from $\bm{DS}-D$.

Pseudo-code of cac is shown in the following \emph{Algorithm 1}.

\begin{algorithm}[!htb]
	\label{a_approximation}
	\SetKwBlock{DoWhile}{Do}{end}
	\caption{cac($\bm{DS},D$)}
	\KwIn{$\bm{DS}=\{x_i\in\mathbb{R}^d|i=1,\cdots,n\},D\subset\bm{DS}$;}
	\KwOut{$C=\cap_{k=1}^K\{z\in\mathbb{R}^d|\alpha_k\cdot z+\gamma_k\leq 0\}$ is a convex area or $C=\emptyset$}
	$C\leftarrow \mathbb{R}^d, U\leftarrow \bm{DS}-D$ \;
	\For{$x\in U$}{
		$\pi \leftarrow$ gslp$(x, D)$ \;
		\eIf{$\pi=\emptyset$}{\Return{$\emptyset$}\;}{\textbf{continue}\;}
		\eIf{$\pi(x)>0$}{$\pi_c(z)\leftarrow -\pi(z)$ \;}{$\pi_c(z)\leftarrow \pi(z)$ \;}
		$C\leftarrow C \cap \{z\in\mathbb{R}^d|\pi_c(z)\leq 0\}$ \;
	}
	\Return $C$\;
\end{algorithm}

In \emph{Algorithm 1}, gslp$(x,D)$ denotes any $O(n')$ time algorithm that solves the soft-margin support vector machine problem of $D$\cite{jakkula2006tutorial,nocedal2006quadratic}.

The time complexity and correctness is formally described as the following lemma without proof due to the limited space. The details of \emph{Algorithm 1} and \emph{Lemma 2} are shown in Appendix B.

\begin{lemma}
	(1)If $cac(\bm{DS},D)=\emptyset$, $D$ is not convex-area separable in $\bm{DS}$.
	
	(2)If $cac(\bm{DS},D)=C\neq\emptyset$, $D$ is separated from $\bm{DS}-D$ by $C$.
	
	(3)The time complexity of cac is $O(nn')$ when $|D|=n'$ and $|\bm{DS}|=n$.
\end{lemma}

By means of the algorithm cac, a naive algorithm for the original CALR problem with $M=1$ is given as \textit{Algorithm 4}, and the time complexity is $O(n2^nd^3)$. The detailed explanation and analysis of \textit{Algorithm 4} is shown in Appendix C due to the limited space. 

\subsection{Algorithms of CALR}
This section discusses a simpler case of CALR problem. Firstly, we give the following definition of convex-area separable for datasets.

\begin{definition}[Convex-Area Separable]
	Given $D=\{(x_i,y_i)|x_i\in\mathbb{R}^d,y_i\in\mathbb{R},i=1,\cdots,n\}$, if there exists a CALR function $F_H$ with $|H|=M+1$ satisfying $\forall x\in S_i,|f_i(x)-y|<\varepsilon$ for $i=1,\cdots,M+1$, where $\varepsilon>0$ is small enough, and $\Vert f_i-f_j\Vert\geq \delta$, $D$ is said to be $(M,\varepsilon,\delta)$-Convex-Area Separable.
\end{definition}

Intuitively, if $D$ is Convex-Area Separable, $D$ could be separated into $M+1$ subsets and $M$ of them could be bounded in disjoint convex areas, and different subsets fit different linear functions.
Such assumption is used in data description tasks and regression modelling for datasets with small errors.
The \textit{Algorithm 2} is designed for Convex-Area Separable datasets for any given $M$. The sub-algorithms used in Line 17 and 26 are shown in \emph{Algorithm 3} and \emph{Algorithm 4} respectively.

\begin{algorithm}[!h]
	\label{a_approximation}
	\SetKwBlock{DoWhile}{Do}{end}
	\caption{casCALR($\bm{DS},M,\tau,\varepsilon,\delta$)}
	\KwIn{$\bm{DS}=\{(x_i,y_i)|x_i\in\mathbb{R}^d,y_i\in\mathbb{R},i=1,\cdots,n\}$ Convex-Area Separable;}
	\KwOut{$\bm{H}=\{(f_i,C_i)|i=0,\cdots,M\}$ such that $F_H \in CALF$, minimizing $\sum_{i=1}^n (F_H(x_i)-y_i)^2$}
	$\bm{H}=\emptyset, F=\emptyset, DS'\leftarrow \bm{DS}$ \;
	\textbf{While} ({$|F|<M$ or $DS'\neq \emptyset$})
	\DoWhile{
		$D \leftarrow$ uniformly sample $d+1$ data points from $DS'$, $f_t \leftarrow$ lr($D$)\;
		\eIf{$p_F(f_t)\geq\tau$}{\textbf{continue}\;}{
			$C_t \leftarrow$ cac($DS',D$)\;
			\eIf{$C_t = \emptyset$}{\textbf {continue}\;}{
				$m\leftarrow |\{f\in F|\Vert f-f_t \Vert\geq\delta\}|$ \;
				\eIf{$m<|F|$}{\textbf{continue}\;}{$F\leftarrow F\cup\{f_t\}$ \;
					$DS'\leftarrow \{(x,y)\in DS'| |y-f(x)|\geq\varepsilon\}$\;}
			}
		}
	}
	$DS'\leftarrow$ distinct$(F,\bm{DS})$ \;
	\For{$f\in F$}{
		$D_f=\{(x,y)\in DS'||y-f(x)|<\varepsilon\}$, $C\leftarrow$ cac($DS',D_f$)\;
		\eIf{$D - C\cap DS' \neq \emptyset$}{
			$f_0 \leftarrow f$ \;
		}{
			$DS' \leftarrow DS' - DS'\cap C$\;
			$\bm{H}\leftarrow \bm{H} \cup \{(f,C)\}$ \;
		}
	}
	$DS' \leftarrow DS' - \{(x,y)\in DS'||y-f_0(x)|<\epsilon\}$ \;
	$\bm{H} \leftarrow \bm{H} \cup post(\bm{H},DS',\varepsilon) \cup \{(f_0,\mathbb{R}^d-\cup_C\{x\in C|(f,C)\in\bm{H}\})\}$ \;
	\Return $\bm{H}$\;
\end{algorithm}

As shown by the pseudo-code, casCALR firstly samples subsets $D$ which can be separated from $\bm{DS}-D$ by hyperplanes and $|D|=d+1$, to find the $M$ different functions $f_1,\cdots,f_M$ in Line 2-16.
In Line 17, the algorithm invoked distinct($F,\bm{DS}$) to remove the data points fitting more than one functions in $F$.
Next in Line 18-25, casCALR finds all data points fitting $f_i$ for $i=1,\cdots,M$, and uses the data points to construct the convex area $C$ for $f_i$. Now there's only a few data points remaining. 

Finally, casCALR invokes a sub-algorithm post$(\bm{H},DS',\varepsilon)$ to add the left data points into output $\bm{H}$ in Line 26-27.
The correctness of casCALR is formally shown in \textit{Theorem 2}. The proof of \textit{Theorem 2} and detailed explanation of casCALR is shown in Appendix D.

\begin{algorithm}[!htb]
	\label{a_approximation}
	\SetKwBlock{DoWhile}{Do}{end}
	\caption{distinct($F,\bm{DS}$)}
	\KwIn{$\bm{DS}=\{(x_i,y_i)|x_i\in\mathbb{R}^d,i=1,\cdots,n\},F={f_i|1\leq i \leq M}$;}
	\KwOut{$DS'\subset \bm{DS}$}
	$DS' \leftarrow \emptyset, D \leftarrow {d=(x,y,0)|(x,y)\in\bm{DS}}$\;
	\For{$f\in F$}{
		\For{$d\in D$}{
			\eIf{$|y-f(x)|<\epsilon$}{
				\eIf{$d=(x,y,0)$}{$ d \leftarrow (x,y,1)$ \;}{$ d \leftarrow (x,y,-1)$ \;}
			}{\textbf{continue} \;}
		}
	}
	$DS' \leftarrow {(x,y)|(x,y,z)\in D, z=1}$ \;
	\Return $DS'$\;
\end{algorithm}

\begin{algorithm}[!htb]
	\label{a_approximation}
	\SetKwBlock{DoWhile}{Do}{end}
	\caption{post($\bm{H},D,\varepsilon$)}
	\KwIn{$\bm{H}=\{(f_i,H_i)|i=1,\cdots,M\},D\subset\bm{DS}$;}
	\KwOut{$\bm{H}'=\{(g_i,H'_i)\}$ where $H'_i$ is a polyhedron}
	$F\leftarrow \{f|(f,H)\in\bm{H}\}$ \;
	\For{$f\in F$}{
		$D_1 \leftarrow \{(x,y)\in D||f(x)-y|<\varepsilon\}$ \;
		\For{$f_t\in F$ and $f_t\neq f$}{
			$D_2 \leftarrow \{(x,y)\in D||f_t(x)-y|<\varepsilon\}$ \;
			\eIf{$D_1\cap D_2 = \emptyset$}{\textbf{continue}\;}{
				$C_t \leftarrow$ cac($D, D_1 \cap D_2$) \;
				$\bm{H}'\leftarrow \{(f,C_t)\}$\;
				$D \leftarrow D - D_1\cap D_2$\;}
		}
	}
	\Return $\bm{H}'$\;
\end{algorithm}

\begin{theorem}
	(1)If $DS'=$distinct$(F,\bm{DS})$, and $D_0',\cdots,D_M'$ are subsets fitting $f_0,f_1,\cdots,f_M\in F$ respectively, there exists at most one $D_i'$ can not be separated from $DS'-D_i'$ by a convex area.
	
	(2)After Line 26 of casCALR, $\bm{H}=\{(f_i,C_i)|i=1,\cdots, M\}$, where $C_i$s are disjoint convex areas bounded by hyperplanes, and $\{(x,y)\in \bm{DS}||f_0(x)-y|<\varepsilon\}$ might not be separated from $\bm{DS}$ by hyperplanes.
\end{theorem}

\subsection{Analysis of Time Complexity}
This section discusses the time complexity of casCALR.
Firstly, the following Lemma 3 is given to analyze the expected time complexity.

\begin{lemma}
	Suppose that a set $D$ is separated in $M$ disjoint subsets $D_1,\cdots,D_M$, where $|D|=n$ and $|D_i|=N_i\geq d+1, i=1,\cdots,M$. Let $D'_1,D'_2,\cdots,D'_t,\cdots$ be a uniformly sampled series of subsets, such that $D'_t\subset D$ and $|D'_t|=d+1$. Let event $A$ be $D_t'\subset D_i$ for some $i$. Let $T$ be the smallest $t$ such that $A$ happens, then $\mathbb{E}[T]=O(q^{d+1}M^{d+1})$, where $\frac{1}{q}=1-\frac{Md}{n}$. 
\end{lemma}
\begin{proof}
	For any $t$, $P(A)=p=\frac{C_{N_1}^{d+1}+\cdots+C_{N_M}^{d+1}}{C_n^{d+1}}$. As the definition of expectation, $\mathbb{E}[T]=p+2(1-p)p+\cdots+t(1-p)^{t-1}p+\cdots$. This is the mathematical expectation of geometry distribution\cite{bertsekas2008introduction}. So $\mathbb{E}[T]=\frac{1}{p}$.
	
	Thus, $\mathbb{E}[T]=\frac{C_n^{d+1}}{C_{N_1}^{d+1}+\cdots+C_{N_M}^{d+1}}=\frac{n(n-1)\cdots (n-d)}{N_1(N_1-1)\cdots(N_1-d)+\cdots+N_M(N_M-1)\cdots(N_M-d)}$. Let $N=\max\{N_1,\cdots,N_M\}$, and reduce the denominator, then there is $\mathbb{E}[T]\leq \frac{n(n-1)\cdots (n-d)}{N(N-1)\cdots(N-d)}$. Furthermore, $\frac{n(n-1)\cdots (n-d)}{N(N-1)\cdots(N-d)}\leq \frac{n^{d+1}}{N(N-1)\cdots(N-d)}\leq \frac{n^{d+1}}{(N-d)^{d+1}}=(\frac{n}{N-d})^{d-1}$.
	
	By the pigeonhole principle, $N\geq \frac{n}{M}$ since $N$ is the biggest $N_i$. So $\frac{n}{N-d}\leq \frac{n}{n/M-d}=\frac{M}{1-Md/n}$. Let $\frac{1}{q}=1-\frac{Md}{n}$, then $Md=n(1-\frac{1}{q})$ and $\mathbb{E}[T]\leq (qM)^{d+1}$.
	
	Therefore, $\mathbb{E}[T]=O((qM)^{d+1})$. Besides, $\mathbb{E}[T]=O(2^{d+1}M^{d+1})$ when $n\geq 2Md$, which is always satisfied in practice.
\qed\end{proof}

The expected time complexity of CALR is given in the following Theorem 3.

\begin{theorem}
	The expected time complexity of casCALR is $O(q^{d+1}M^{d+2}nd+M^2nn')$, where $\frac{1}{q}=1-\frac{(M+1)d}{n}$ and $n'=max_i |C_i \cap \bm{DS}|$.
\end{theorem}
\begin{proof}
	From the pseudo-code of casCALR, it firstly finds every $f_i$, by randomly sampled subsets $D$ with $|D|=d+1$. Suppose the biggest expected time of sampling $D$ for $f_i$ is $T$. For each suitable $D$, casCALR costs $O(d^3)$ to construct linear regression model on $D$, $O(n)$ time to construct $D'$ and $U'$ by checking the predicting error of every $(x,y)\in \bm{DS}$, $O(n(d+1))$ time to construct the convex area containing $D$, $O(M)$ time to test whether $f_i$ is contained in $F$, $O(n)$ time to remove the data points fitted by $f_i$. Thus, the time complexity of Line 3-16 for each $D$ is $O(d^3+n+n(d+1)+M+n)=O(nd+d^3)$. These steps will be repeated for $M$ times to find all $f_i$ in $\bm{H}$. Suppose that $T$ is the max sample time of constructing every $f_i$, the time complexity of Line 2-16 is $O(MT(nd+d^3))$.
	
	Next, casCALR finds $C_i$ for each $f_i\in F$. Firstly casCALR invoked the sub-algorithm distinct$(F,\bm{DS})$ to find all data points fitted by $f_i$, and removes data points fitted by $f_j, j\neq i$ to make sure the convex area $C_i$ containing them is rightly constructed. From the pseudo-code of Algorithm 3, it's obvious that the time complexity of Line 17 is $O(Mn)$.
	
	Then casCALR costs $O(nn')$ time to construct convex area $C_i$ where $n'=|D|$, and $O(n)$ time to remove the data points fitted by $f_i$. The total time complexity of Line 18-25 is $O(M(nn'+n))=O(Mnn')$.
	
	In the end, casCALR deals with the data points fitted by more than one $f\in F$. The sub-algorithm post($\bm{H},DS',\varepsilon$) is invoked, and the time complexity of it is $O(M^2nn')$ from Appendix D.
	
	Adding the four parts of the algorithm, the total time complexity of casCALR is $O(MT(nd+d^3)+Mn+Mnn'+M^2nn')=O(MT(nd+d^3)+M^2nn')$. It's known that $\mathbb{E}[T]=q^{d+1}(M+1)^{d+1}$ where $\frac{1}{q}=1-\frac{(M+1)d}{n}$ from Lemma 3. Therefore, the final expected time complexity of casCALR is $O(q^{d+1}M^{d+2}(nd+d^3)+M^2nn')$.
\qed\end{proof}

Noticing that $Mn'=o(n)$ since $C_0$ contains most data points of $D$. Therefore, generally the expected time complexity of casCALR is $o(Mn^2)$ when $d$ is not very large.
Besides, a much more concise cas2 in \textit{Algorithm 7} can be proposed if $\bm{DS}$ is convex-area separable and $M=1$ in Appendix E. The expected time complexity of cas2 is $o(2^{d+1}d^3+n^2)$.

Besides, there are several methods to estimate good enough $\varepsilon, \tau$ and $\delta$ if they are not given. 
In practice, $\tau=0.05$ and $\delta=0.5$ is an accurate enough threshold for hypothesis testing\cite{book2021lra}.
Once a linear regression model $f$ on $D$ satisfies $p_f<0.05$, one can have $\varepsilon=3\sqrt{\frac{1}{|D|}\sum_{(x,y)\in D}(f(x)-y)^2}$ as an accurate estimator.
Therefore, $\varepsilon, \tau$ and $\delta$ is not necessary to be given for casCALR.

\section{Conclusion and Future Work}
This paper introduces a new function class called Convex-area-wise Linear functions(CALF), which has high interpretability and strong expressive ability. This paper proves that CALF can interpolate any given finite datasets, and can represent discontinuous piece-wise linear functions. This paper formally proposed convex-area-wise linear regression(CALR) problem for data analysis and prediction, along with an optimizing version of it to approximately solve the problem using existing optimization solvers.

Besides, this paper proposed several algorithms to solve CALR problem under different assumptions. 
If assuming the given datasets are convex-area separable, an $O(q^{d+1}M^{d+2}nd+M^2nn')$ expected time algorithm casCALR is proposed for the case when $M$ is known, and an $o(2^{d+1}d^3+n^2)$ expected time algorithm cas2 is proposed when $M=1$. Both casCALR and cas2 are accurate algorithms for the problem. Since $2^{d+1}d^3, q^{d+1}M^{d+2}nd, Mn' < n^2$ when $d$ is relatively small, the time complexity of them can reach almost $o(Mn^2)$ in practice.

For future work and challenge of CALF and CALR problem, reduction of time complexity is the most important part. $O(q^{d+1}M^{d+2}nd+M^2nn')$ is lower than some currently existing regression methods such as PLDC\cite{siahkamari2020piecewise} when $d$ and $M$ are relatively small, but it's still high for common big datasets. A linear time, even sub-linear time approximate algorithms for CALR is urgently needed. Secondly, an efficient algorithm for CALR problem when datasets are not convex-area-separable is also a challenge.
Besides, a searching algorithm for best $M$ is also needed to mitigate over-fitting when $M$ is unknown. Taking each data point in given $\bm{DS}$ as a local model can reach zero error, but it causes serious over-fitting to the regression model, which makes it important to adjust $M$ flexibly.

%
%
%
\bibliographystyle{splncs04}
\bibliography{reference}

\clearpage
\section*{Appendix}
\appendix

\section{Illustrations of CALF}
The following Figure 1 and 2 give an illustration of multiple-model linear regression and CALF, where $x=(x_1,x_2)$ and $y=f(x)=f(x_1,x_2)$.

 \begin{figure}
 	\centering
 	\includegraphics[width=0.9\linewidth]{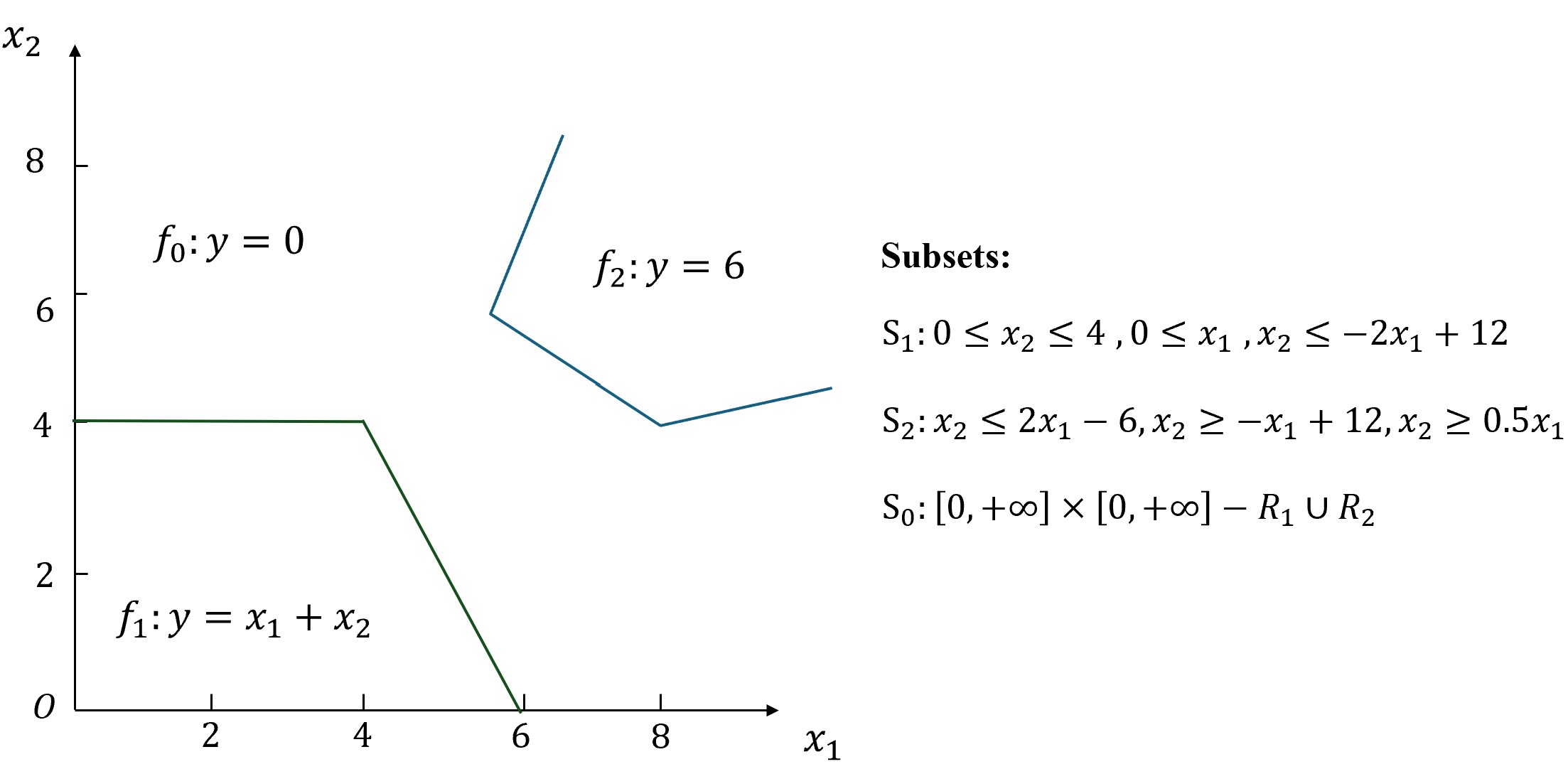}
 	\caption{An illustration of multiple-model linear regression. Modelling $y$ with explain variables $(x_1,x_2)$, three local models with three subsets used to best fit the dataset.}
 	\label{fig:DPRV}
 \end{figure}

As shown in Figure 1, given a 3-dimensional $\bm{DS}=\{(y, x_1 , x_2)\} \subset \mathbb{R}\times [0, +\infty] \times [0, +\infty]$, we can use three linear regression models, $y=x_1+x_2$ for $S_1= \{ (x_1,x_2) | 0\leq x_2\leq 4,0\leq x_1,x_2\leq -2x_1+12  \}$, $y=6$ for $S_2 = \{(x_1,x_2)|x_2\leq 2x_1-6,x_2\geq -x_1+12,x_2\geq 0.5x_1\}$ and $y=0$ for $S_3 = [0, +\infty] \times [0, +\infty]- S_1\cup S_2 $, to accurately model $\bm{DS}$. 
And the final multiple-model linear regression model can be represented by $f(x)=I_{S_1}(x_1+x_2)+6I_{S_2}+0I_{S_0}$.
Since the three models are very different, it is impossible to model $\bm{DS}$ using any single model as accurate as the three models do.

The Figure 2 gives an example of convex-area-wise linear function. The partition of $(x_1,x_2)$ is the same as in Figure 1. So there is $H=\{(f_1,S_1),(f_2,S_2),(f_0,S_0)\}$, where $f_1(x)=x_1+x_2$, $f_2(x)=6$, $f_0(x)=0$. And $S_1 = \{-x_1\leq 0\}\cap\{-x_2\leq 0\}\cap\{x_2\leq 4\}\cap \{2x_1+x_2-12\leq 0\}$, $S_2 = \{-2x_1+x_2+6\leq 0\}\cap\{-x_1-x_2+12\leq 0\}\cap\{\frac{1}{2}x_1+x_2\leq 0\}$, $S_0=\mathbb{R}^d-S_1\cup S_2$.
Obviously, $S_1$ is bounded by 4 lines and $S_2$ is bounded by 3 lines and $S_1\cap S_2 \cap S_0=\emptyset$. Therefore, $F_H=I_{S_1}(x_1+x_2)+6I_{s_2}$ is a convex-area-wise linear function.

\begin{figure}
	\centering
	\includegraphics[width=0.9\linewidth]{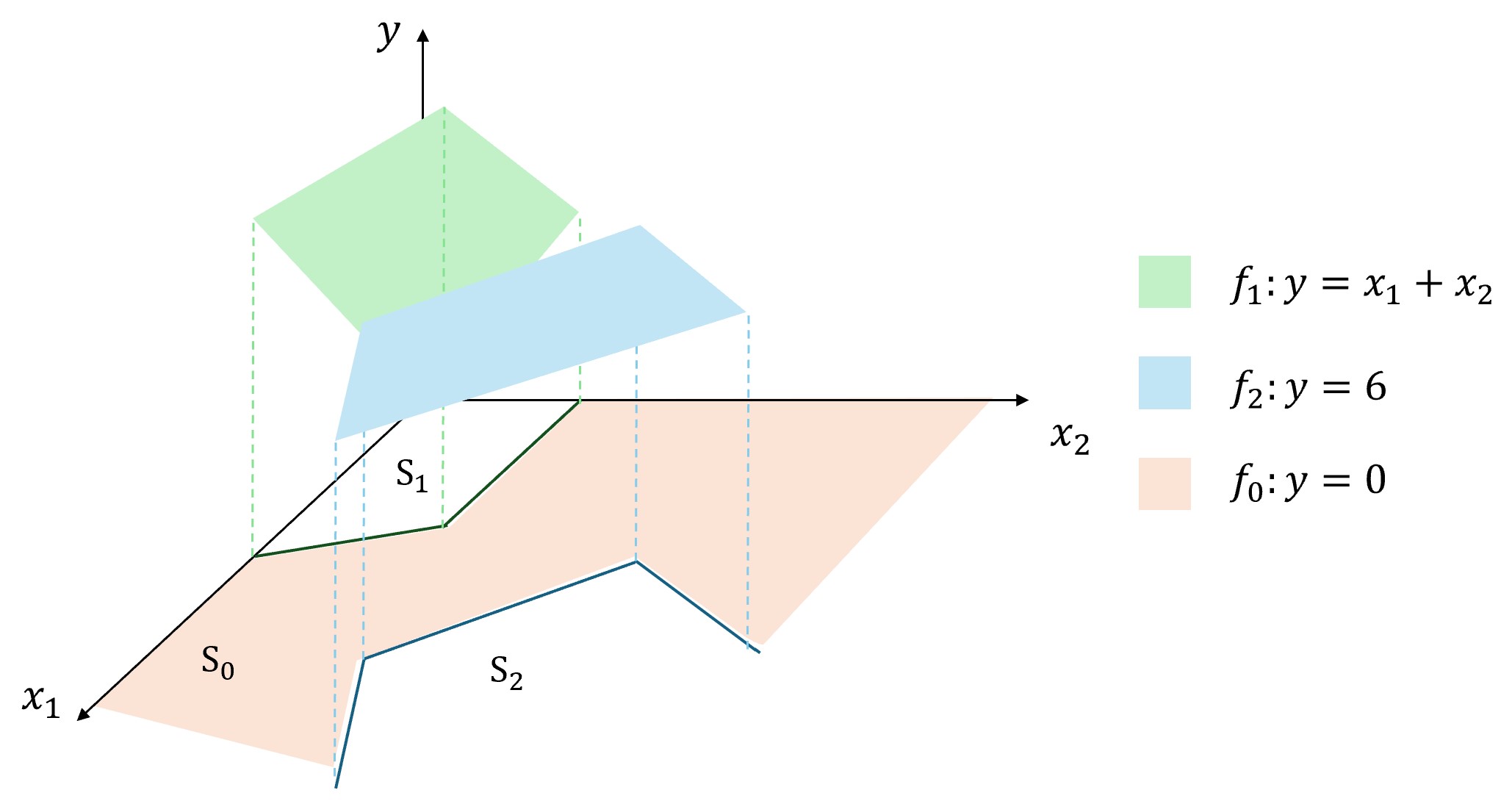}
	\caption{An illustration of CALF, with $H=\{(f_1,S_1),(f_2,S_2),(f_0,S_0)\}$.}
	\label{fig:DPRV}
\end{figure}

\section{Algorithm for finding convex area containing $D$}
The \emph{Algorithm 1} in Section 3 can separate $D$ from $\bm{DS}-D$ by a convex area, if and only if the convex hull of $D$ contains no points of $\bm{DS}-D$. The algorithm is efficient when $d$ is not large. For further discussion of \emph{Algorithm 1}, we firstly consider a similar method using support vector machine rather than linear programming.

Consider the following soft-margin Support Vector Machine Problem without kernel methods(SVM): given datasets $D=\{(x_i,y_i)|x_i \in \mathbb{R}^d, y_i\in\{1,-1\}\}$, return an hyperplane $\pi:\alpha\cdot x+\gamma=0$ where $\alpha\in\mathbb{R}^d,\gamma\in\mathbb{R}$, maximizing the soft margin between class $y=1$ and $y=-1$. This soft-margin SVM is equivalent to the following problem\cite{nocedal2006quadratic,jakkula2006tutorial}.

\begin{definition}[soft-margin SVM problem]
	$$\min_{\alpha,\gamma,\xi} \frac{1}{2}\Vert \alpha \Vert + C \sum_{i=1}^n\xi_i$$
	$$s.t. \begin{cases}
		y_i(\alpha\cdot x_i + \gamma)\geq 1-\xi_i \\
		\xi_i \geq 0 \\
		i=1,\cdots,n\\
	\end{cases} $$
	where $x,\alpha\in\mathbb{R}^d,\gamma \in \mathbb{R}, C>0$ is a given hyper-parameter.
\end{definition}

It's a linear constraint quadratic programming, with $n+d+1$ optimizing variables and $2n$ constraints. 
And it can be solved in $O((n+d+1)^3)$ time using conjugate gradient or Cholesky splitting methods\cite{nocedal2006quadratic}.

We denote $svm(D)$ as any $O((n+d+1)^3)$ time algorithm solving SVM problem, taking $D$ as input and a hyperplane $\pi$ as output. Besides, it's easy to make sure no $x\in D$ such that $\pi(x)=0$. If so, substitute $\pi(x)$ with $\pi'(x)=\pi(x)+\epsilon$ can solve it, where $|\epsilon|$ is small enough. Therefore, we always assume that $\pi=svm(D)$ satisfying that no $x\in D$ such that $\pi(x)=0$.

As for the output $\pi$ of the problem, $\pi$ perfectly separates $D$ into $D_1=\{(x,y)\in D|y=1\}$ and $D_2=\{(x,y)\in D|y=-1\}$ if $D$ is linear separable. 
Let $\pi(x)=\alpha\cdot x+\gamma$, there must exists $(x_1,1),(x_2,-1)\in D$ such that $\pi(x_1)\cdot \pi(x_2)>0$, if $D$ is not linear separable.

Taking $svm(D)$ as a sub-algorithm, A nearly $O(n)$ time algorithm can be proposed to test whether a subset $S\subset D$ is convex-area separable in $D$ without constructing convex hull of $S$, which is described as the following \textbf{Algorithm 5}.

\begin{algorithm}[!htb]
	\label{a_approximation}
	\SetKwBlock{DoWhile}{Do}{end}
	\caption{cacs($\bm{DS},D$)}
	\KwIn{$\bm{DS}=\{x_i\in\mathbb{R}^d|i=1,\cdots,n\},D\subset\bm{DS}$;}
	\KwOut{$C=\cap_{k=1}^K\{z\in\mathbb{R}^d|\alpha_k\cdot z+\gamma_k\leq 0\}$ is a convex area or $C=\emptyset$}
	$C\leftarrow \mathbb{R}^d, D'\leftarrow \{(x,1)|x\in D\}, U'\leftarrow \{(x,-1)|x\in \bm{DS}-D\}$ \;
	\For{$(x,-1)\in U'$}{
		$\pi \leftarrow$ svm$(D'\cup\{(x,-1)\})$ \;
		\For{$(x',1)\in D'$}{
			\eIf{$\pi(x')\pi(x)>0$}{\Return{$\emptyset$}\;}{\textbf{continue}\;}
		}
		\eIf{$\pi(x)>0$}{$\pi_c(z)\leftarrow -\pi(z)$ \;}{$\pi_c(z)\leftarrow \pi(z)$ \;}
		$C\leftarrow C \cap \{z\in\mathbb{R}^d|\pi_c(z)\leq 0\}$ \;
	}
	\Return $C$\;
\end{algorithm}

The algorithm cacs firstly marks the given $D$ as class $1$ and $\bm{DS}-D$ as class $-1$, denoted as $D'$ and $U'$. Then in line 2-13, cacs uses svm($D'\cup\{(x,-1)\}$) for any $(x,-1)\in U'$ to find a hyperplane $\pi$ to best separating $D'$ with $\{(x,-1)\}$. After getting $\pi$, cacs tests whether $\pi$ could truly separate $D'$ and $\{(x,-1)\}$. If so, cacs adds $\pi_c=\pi$ or $-\pi$ into output $C$, and return $\emptyset$ otherwise in Line 4-13. 
If cacs never returns $\emptyset$ for every element in $U'$, it returns $C$ as a convex area separating $D$ and other data points in $\bm{DS}$. Obviously, the time complexity of cacs is $O(n((n'+d+1)^3+n'))=O(n(n'+d)^3)$ when $|D|=n'$ and $|\bm{DS}|=n$.

In practice, we can add $U'\leftarrow U'-\{(x',-1)\in U'|\pi_c(x')\leq 0\}$ after Line 13 to reduce the number of elements in $U'$ to be tested and reduce the output size $|C|$. It may significantly reduce $C$.

The following Theorem shows the time complexity and correctness of cacs.
\begin{theorem}
	If $cacs(\bm{DS},D)=\emptyset$, $D$ is not convex-area separable in $\bm{DS}$.
	If $cacs(\bm{DS},D)=C\neq\emptyset$, $D$ is separated from $\bm{DS}-D$ by $C$.
	
	The time complexity of cac is $O(n(n'+d)^3)$ when $|D|=n'$ and $|\bm{DS}|=n$.
\end{theorem}
\begin{proof}
	Let $(u_k,-1)$ be the element of $U'$, $k=1,\cdots,n-n'$ and $C_D$ is the convex hull of $D$. As the knowledge in convex analysis shows, $u_k$ can not be separated with $D$ by hyperplane if and only if $u_k\in C_D$. Besides, $u_k \in C_D$ if and only if there exists $x\in D$ such that $\pi(x)\pi(u_k)>0$, where $\pi=svm(D',(u_k,-1))$, as discussed above.
	
	Therefore, $C=\emptyset$ if and only if there exists an $(u_k,-1) \in U'$ such that $u_k \in C_D$, which indicates $D$ can not be separated with $\bm{DS}$ by hyperplanes.
	
	If $C\neq \emptyset$, let $C=\cap_{k=1}^{n-n'}C_k=\cap_{k=1}^{n-n'}C_k\{z\in\mathbb{R}^d|\pi_k(z)\leq0\}$. Obviously, $C$ is the intersection of $K$ semi-spaces, which makes it a convex area bounded by hyperplanes. Besides, $D\subset C_k$ and $u_k \notin C_k$. So no $u_k\in \bm{DS}-D$ such that $u_k \in \cap C_k$, which means $D$ is separated with $\bm{DS}-D$ by $C$. 
\qed\end{proof}

In the proof of \emph{Theorem 4}, the svm could be directly replaced by gslp. The time complexity of gslp($x,D$) is $O(n')$. Therefore the time complexity of \emph{Algorithm 1} in Section 3 is $O(nn')$. Thus, \emph{Lemma 2} is proved. 

\section{Naive Algorithm for CALR problem when $M=1$}
We introduced naiveCALR in Section 3, and the corresponding analysis is shown in this section. The pseudo-code of naiveCALR is shown in the following \emph{Algorithm 6}
 
\begin{algorithm}[!htb]
	\label{a_approximation}
	\SetKwBlock{DoWhile}{Do}{end}
	\caption{naiveCALR($\bm{DS}$)}
	\KwIn{$\bm{DS}=\{(x_i,y_i)|x_i\in\mathbb{R}^d,y_i\in\mathbb{R},i=1,\cdots,n\}$, M=1;}
	\KwOut{$\bm{H}=\{(f_1,S_1),(f_d,S_d)\}$ such that $f_{H}\in CALR$, minimizing $\sum_{i=1}^n (F_H(x_i)-y_i)^2$}
	$N=1, \bm{H}=\emptyset, L=+\infty$ \;
	\textbf{While} ({$N<n-1$})
	\DoWhile{
		\textbf{enumerate} ({$D\subset \bm{DS}$ and $|D|=N$})
		\DoWhile{
			$C \leftarrow$ cac($\bm{DS},D$) \;
			\eIf{$|C\cap\bm{DS}|>|D|$}{\textbf{continue}\;}{
				$f \leftarrow$ lr($D$), $l_1 \leftarrow \sum_{(x,y)\in D}(f(x)-y)^2$ \;
				$f_d \leftarrow$ lr($\bm{DS}-D$), $l_d \leftarrow \sum_{(x,y)\in \bm{DS}-D}(f_d(x)-y)^2$ \;
				\eIf{$l_1 + l_d < L$}{\textbf{continue}\;}{
					$L\leftarrow l_1+l_d$ \;
					$\bm{H}\leftarrow \{(f,C),(f_0,\mathbb{R}^d - C)\}$ \; 
				}
			}
		}
	}
	\Return $\bm{H}$\;
\end{algorithm}

The naiveCALR enumerates all subsets $D$ contains exactly $N$ data points of $\bm{DS}$ for $N=d+1,d+2,\cdots,n-d-1$. 
The algorithm uses cac($\bm{DS},D$) for each $D$, judging whether $D$ could be separated from $\bm{DS}$ by hyperplanes. If cac($\bm{DS},D$)$\neq \emptyset$, $D$ can be a candidate of $\bm{H}$.
For a candidate $D$, the algorithm models $f$ for $D$ and $f_d$ for $\bm{DS}-D$ to get $F_H$. If the new $F_H$ has smaller MSE, algorithm upgrades the solution $\bm{H}$. 
Obviously, naiveCALR enumerates all possible $D$ and output the $F_H$ with smallest MSE.

As for the time complexity, enumerating all possible subsets $D$ costs $C_n^{d+1}+C_n^{d+2}+\cdots+C_n^{n-d-1}=O(2^n)$ time. Calculating cac($\bm{DS},D$) for each $D$ and testing $\bm{DS}-D$ costs $O(n d^3)$ time. Modelling $f,f_d$ and calculating $l_1,l_d$ costs $O(n)$ time. Therefore, the total time complexity of naiveCALR is $O(n2^nd^3)$.

\section{Details of casCALR and cas2}
We firstly explain the correctness for Line 1-25 of casCALR, then we show the correctness and time complexity of the sub-algorithm post($\bm{H},DS'$). Combining the two parts of explanation, the correctness of casCALR is illustrated.

\begin{theorem}
	After Line 25 of casCALR, $\bm{H}=\{(f_i,C_i)|i=1,\cdots, M\}$, where $C_i$ are disjoint convex area bounded by hyperplanes. And $\{(x,y)\in \bm{DS}||f_0(x)-y|<\varepsilon\}$ might not be separated from $\bm{DS}$ by hyperplanes.
\end{theorem}
\begin{proof}
	Consider that $\bm{DS}$ is convex-area separable, and the underlying model is $H=\{(f_i,C_i)|i=0,1,\cdots,M\}$.
	Then for any $D=\subset \bm{DS}$ and $|D|=d+1$, suppose that $f_D$ is the linear regression model of $D$, $f_D$ satisfies $p_F(f_D)<\tau$ if and only if $|f_D(x)-y|<\varepsilon$ for every $(x,y)\in D$, since $\varepsilon$ is small enough\cite{book2021lra}.
	
	By the property of linear function and linear regression, there are two cases that making sure $p_F(f_D)<\tau$\cite{book2021lra}. One is that $(x,y)\in C_i$ where $f_i\approx f_D$, the other is that $(x,y)\in C_j, j\neq i$, but $|f_j(x)-f_i(x)|<\epsilon$. 
	Let $\pi_{ij} = f_j(x)-f_i(x)$, then $x$ is close to the hyperplane $\pi_{ij}$. Then $\pi_{ij}$ is either inside $C_j$ or is $C_i\cap C_j$.
	Whatever the case, $(x,y)$ satisfies $|f_i(x)-y|<\varepsilon$ and $|f_j(x)-y|<\varepsilon$ simultaneously.
	
	Let $D_i = \{(x,y)\in\bm{DS}||f_i(x)-y|<\varepsilon\}$, $D'_i\subset D_i$ and there exists one $j$ such that $|f_j(x)-y|<\varepsilon$. For any $f_i$ and corresponded $\bm{DS}\cap C_i$, removing the data points such that $(x,y)\in C_j$ and $|f_i(x)-y|<\varepsilon$, the leaving data can still be separated from the rest by hyperplanes since $C_i$ is convex. 
	As for $D_{ij}=\{(x,y)\in C_i||\pi_{ij}(x)|<\varepsilon\}$, they are already separated from $D_i-D'_i$ by hyperplane $\pi_{ij}(x)+\varepsilon=0$ and $\pi_{ij}(x)-\varepsilon=0$.
	Therefore, $D_i-D'_i$ can be separated from $\bm{DS}$ by $C_i$, and cac($D_i-D'_i,\bm{DS}$) consists $D_i-D_{ij}$.
	
	As for the Line 21, $C_0$ is not necessary to be convex by the definition of CALF. If $C_0$ in the underlying model is not convex, the convex hull of $C_0\cap \bm{DS}$ must contain $(x,y)\in \bm{DS}-C_0\cap\bm{DS}$. It means that there exists $(x,y)\in \bm{DS}-C_0\cap\bm{DS}$ can not be linearly separated from $C_0 \cap \bm{DS}$, which means algorithm cac($C_0\cap\bm{DS},\bm{DS}-C_0\cap\bm{DS}$) returns $\emptyset$. Besides, $C_0$ is the only element in $\bm{H}$ having such characteristics. Therefore, $f_0$ acquired from Line 21 is the function of $C_0$, and $\{(x,y)\in \bm{DS}||f_0(x)-y|<\varepsilon\}$ might not be separated from $\bm{DS}$ by hyperplanes.
\qed\end{proof}

From the proof we can know that after Line 25 of casCALR, any data points in $DS'$ fitting more than one $f_i$ of the current $\bm{H}$. Then Line 26 is designed to allocate them into polyhedron to make sure the final $F_{\bm{H}}$ is CALF. We firstly give the pseudo-code of algorithm post, and the following theorem formally explain it.

\begin{theorem}
	$F_{\bm{H}'}\in CALF$, where $\bm{H}'$ is the output of post($\bm{H},D,\varepsilon$), and $D=\cup_{i=1}^{|\bm{H}|} (D\cap H'_i)$.
	The time complexity of post($\bm{H},D,\varepsilon$) is $O(M^2nn')$, where $n'=|\bm{DS}-D|$.
\end{theorem}
\begin{proof}
	From the proof of Theorem 5, it's known that $D_{ij}=\{(x,y)||f_i(x)-y|<\varepsilon, |f_j(x)-y|<\varepsilon\}$ is bounded by two hyperplanes $f_i(x)-f_j(x)-\varepsilon=0$ and $f_i(x)-f_j(x)+\varepsilon=0$. So $D_{ij}$ can be separated by cac($\{(x,1)|(x,y)\in D_{ij}\},\{(x,-1)|(x,y)\in D- D_{ij}\}$). Also from the proof of Theorem 5, $D$ only contains data points fitting more than one $f\in F$. Therefore, traversing $i,j=1,\cdots,M$ of $f_i$ and $f_j$ contains all data points in $D$, which makes $D=\cup_{i=1}^{|\bm{H}|} (D\cap H'_i)$.
	
	There are two levels of iterations for $f\in F$. In each iteration, the algorithm takes $O(n)$ time to construct $D_1$ and $D_2$, and takes $O(nn')$ to invoke cac. Finally, the time complexity of post($\bm{H},D,\varepsilon$) is $O(M^2 nn')$.
\qed\end{proof}

Combining the \emph{Theorem 5} and \emph{Theorem 6}, the correctness of the algorithm casCALR has been proved.

A much more concise cas2 in \textit{Algorithm 7} can be proposed if $\bm{DS}$ is convex-area separable and $M=1$.

\begin{algorithm}[!htb]
	\label{a_approximation}
	\SetKwBlock{DoWhile}{Do}{end}
	\caption{cas2($\bm{DS},\tau,\varepsilon,\delta$)}
	\KwIn{$\bm{DS}=\{(x_i,y_i)|x_i\in\mathbb{R}^d,y_i\in\mathbb{R},i=1,\cdots,n\}$ convex-area separable;}
	\KwOut{$\bm{H}=\{(f_i,S_i)|i=0,1\}$ such that $F_H \in CALF$, minimizing $\sum_{i=1}^n (F_H(x_i)-y_i)^2$}
	$\bm{H}=\emptyset$ \;
	\textbf{While} {\textbf{TRUE}}
	\DoWhile{
		$D\leftarrow$ uniformly sample $d+1$ data points from $\bm{DS}$, $f_1 \leftarrow$ lr($D$) \;
		\eIf{$p_F(f_1)\geq\tau$}{\textbf{continue}\;}{\textbf{break}\;}
	}
	$D \leftarrow \{(x,y)|(x,y)\in\bm{DS},|f_1(x)-y|< \varepsilon\}$ \;
	$D' \leftarrow \bm{DS}-D$, $f_2\leftarrow$ lr($D'$) \;
	$D_0 = D\cap D'$, $D\leftarrow D-D_0$, $D'\leftarrow D'-D_0$ \;
	$C_1 \leftarrow$ cac$(\bm{DS}-D_0,D)$, $C_2 \leftarrow $cac$(\bm{DS}-D_0,D')$ \;
	\eIf{$C_1 = \emptyset$}{
		$\bm{H}\leftarrow \{(f_2,C_2),(f_1,\mathbb{R}^d-C_2)\}$ \;
	}{$\bm{H}\leftarrow \{(f_1,C_1),(f_2,\mathbb{R}^d-C_1)\}$\;}
	\Return $\bm{H}$\;
\end{algorithm}

Similarly, cas2 samples subsets $D$ to construct $f_1,f_2$ in Line 2-9. Then cas2 finds data points fitting $f_1,f_2$ and constructs $C_1,C_2$ for them in Line 10-15. Differently, cas2 doesn't need to invoke post$(\bm{H},DS',\varepsilon)$. The expected time complexity of cas2 is shown in the following \emph{Theorem 7}.

\begin{theorem}
	The expected time complexity of cas2 is $O(2^{d+1}d^3+nn')$, where $n'=max\{|C_1\cap\bm{DS}|,|C_2\cap\bm{DS}|\}$.
\end{theorem}
\begin{proof}
	From the pseudo-code, cas2 sampled $D$ with $|D|=d+1$ to construct $f_1$ or $f_0$. Once cas2 find a subset $D$ which the linear regression model of it $f$ satisfying $p_F{f}<\tau$, cas2 moves to Line 8. Suppose cas2 samples $T$ times to find the suitable $D$, it costs $O(d^3)$ to construct $f$. So the time complexity of Line 2-7 is $O(Td^3)$.
	
	In Line 8, cas2 directly gets the other $f'\in\bm{H}$ constructing from data points which are not fitted by $f$. Then cas2 finds the subsets fitted by $f$ and $f'$, and removes the data points fitted by both, getting $D_1$ and $D_2$. After that cas2 constructs convex area $C_1$ and $C_2$ containing $D$ and $D'$ respectively. By the assumption, either $D_1\cap C_2=\emptyset$ and $D_2\subset C_1=\emptyset$, or one $C_i$ contains data points from the other subsets. Taking the one $C_i$ such that $C_i \cap D_j\neq\emptyset,i\neq j$ as $C_0$ solves the problem. As discussed, the time complexity of Line 8-15 is $O(n+d^3+nn'+n)=O(nn')$.
	
	Adding the two parts of cas2, the total time complexity is $O(Td^3+nn')$. It's known that $\mathbb{E}[T]\leq (\frac{n}{N-d})^{d+1}$ where $N=\max\{N_1,N_2\}$ from the proof of Lemma 3. Then $\frac{n}{N-d}\leq 2$ since $N\geq \frac{n}{2}$. Therefore, the final expected time complexity of cas2 is $O(2^{d+1}d^3+nn')$.
	
	\qed\end{proof}
Finally, the correctness of cas2 is formally described in the following theorem.
\begin{theorem}
	Let $C_1$ and $C_2$ be constructed in Line 11 of cas2, then if $C_1\neq\emptyset$, $C_1 \cap\bm{DS}$ fitting $f_1$ and $(\mathbb{R}^d - C_1) \cap \bm{DS}$ fitting $f_2$ , and vice versa.
\end{theorem}
\begin{proof}
	Suppose that the underlying model of $\bm{DS}$ is $\bm{H}=\{(g_1,S_1),(g_0,S_0)\}$ and $f_1\cong g_1,f_2\cong g_0$. Thus, $D = C_1\cap (\bm{DS}-D_0) \subset S_1 \cap \bm{DS}$. By the definition of cac, $C_1\neq \emptyset$ and $C_1 \cap \bm{DS} \subset S_1 \cap \bm{DS}$. For the data points such that $x \in S_1 \cap \bm{DS} - C_1 \cap \bm{DS}$, they must satisfy $|f_2(x)-y|<\varepsilon$. So they fits $f_2$ as well. Thus, $D' = C' \cap \bm{DS}$ fits $f_2$ where $C'=\mathbb{R}^d - C_1$.
\qed\end{proof}

\section{Details of pseudo-inverse matrix method}

This section shows the detail of \emph{pseudo-inverse matrix method}.

Given dataset $\bm{DS}$, let $\bm{y} = (y_1, y_2,\cdots,y_n)^T$ and $\bm{X}_{n \times (k+1)}$ be the data matrix of $\bm{DS}$:
$$\bm{X}=\begin{pmatrix}
	1 & x_{11} & x_{12} & \cdots & x_{1d} \\
	1 & x_{21} & x_{22} & \cdots & x_{2d} \\
	\vdots & \vdots & \vdots & \vdots \\
	1 & x_{n1} & x_{n2} & \cdots & x_{nd} \\
\end{pmatrix} .$$ 
The $x_{ij}$ in $\bm{X}$ equals to the value of the $i$-th data point's $j$-th dimension in $\bm{DS}$. Then, using the formula $\hat{\bm{\beta}}=(\bm{X'}\bm{X})^{-1}\bm{X'}\bm{y}$, the linear regression model $\hat{f}$ could be constructed.
The time complexity of pseudo-inverse matrix method is $O(d^2 n + d^3)$ \cite{book2021lra}. When $d$ is big enough, gradient methods is more efficient than pseudo-inverse matrix method. Generally, every method's complexity has the bound $O(d^2 n + d^3)$, so this paper use it as the time complexity of linear regression in common. And let lr($\bm{D}$) be any $O(d^2 n + d^3)$ time algorithm to construct the linear function of given datasets $\bm{DS}$. 

\section{Details of TBI dataset}
Fig.3 is an example of TBI dataset, given in \cite{TBI2005}, which is used to predict traumatic brain injury patients' response with sixteen explanatory variables $x_i , 1 \leq i \leq 16$. In \cite{dong2015pattern}, researchers analyzed this data with both one linear model and several. As shown in Fig.3, the $RMSE$ is 10.45 when one linear regression model is used to model the whole TBI. However the $RMSE$ is reduced to 3.51 while TBI is divided into 7 subsets and 7 different linear regression models are used to model the 7 subsets individually. 
The \emph{goodness of fit}($R^2$ for short) of the one linear regression model for modelling TBI is 0.29. But when using 7 linear regression models to model TBI, $R^2$ increases to $0.85$. 

\begin{figure}[htbp]
	\centering
	\includegraphics[width=0.8\linewidth]{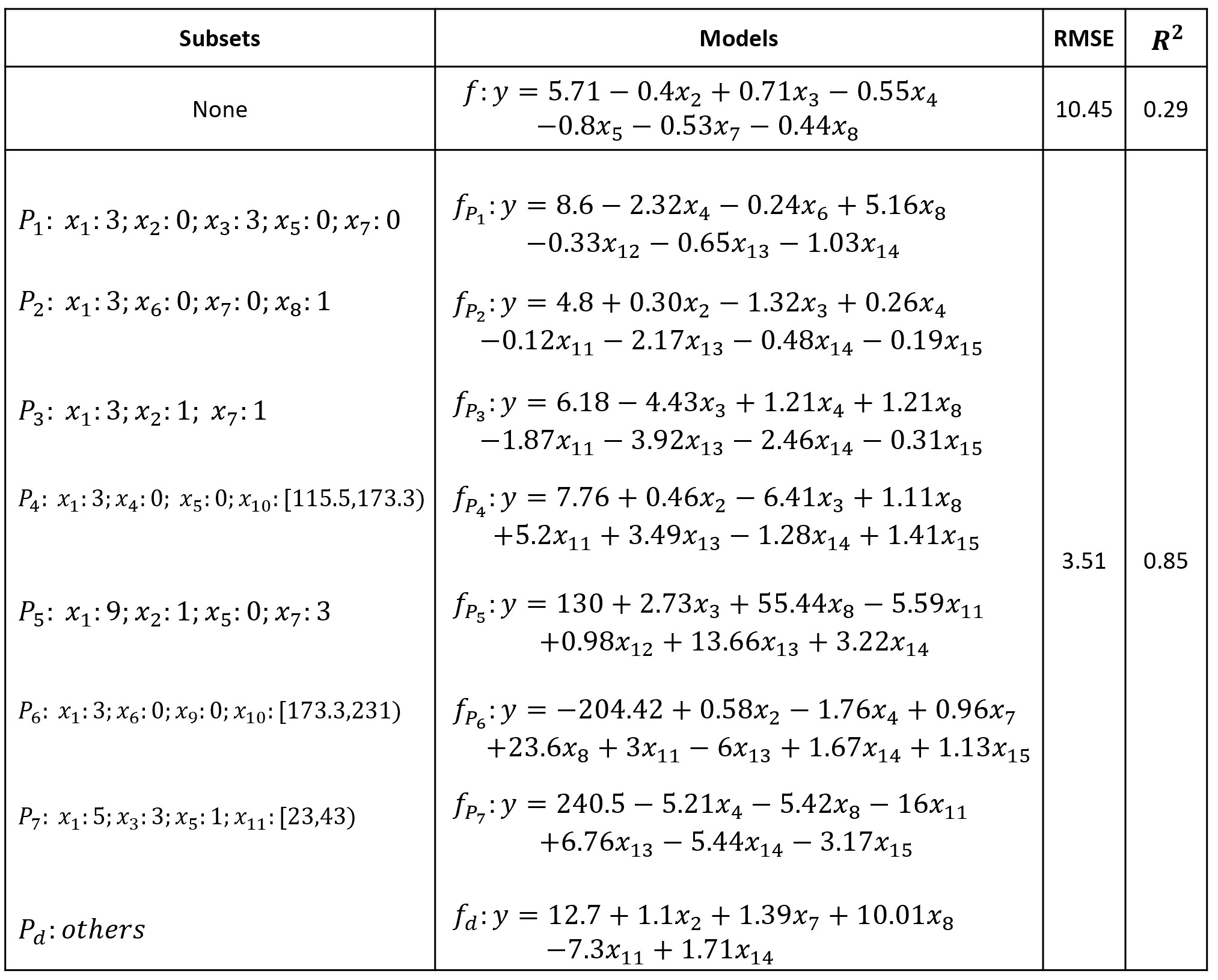}
	\caption{Using one linear regression model and 7 regression models for TBI data, has totally different prediction accuracy.}
	\label{fig:DPRV}
\end{figure}

\end{document}